\documentclass[11p,reqno]{amsart}

\usepackage[T1]{fontenc}
\usepackage{graphicx}
\usepackage{amssymb,amsthm,amsmath,mathrsfs,bm,braket,marginnote,physics}
\usepackage{enumerate}
\usepackage{appendix}
\usepackage{extarrows}
\usepackage[colorlinks=true, pdfstartview=FitV, linkcolor=blue, citecolor=blue, urlcolor=blue]{hyperref}
\usepackage{multirow}
\usepackage{xcolor}

\usepackage{pgf}
\usepackage{pgfplots}
\usepackage{tikz}
\usetikzlibrary{arrows,calc}
\usepackage{verbatim}
\usetikzlibrary{decorations.pathreplacing,decorations.pathmorphing}
\usepackage[numbers,sort&compress]{natbib}
\usepackage{dsfont}

\numberwithin{equation}{section}
\newtheorem{theorem}{Theorem}[section]

\newtheorem{proposition}[theorem]{Proposition}

\theoremstyle{remark}

 \reversemarginpar
 \newcommand{\al}{\alpha}
 \newcommand{\T}{\theta}

  \newcommand{\tp}{p}
    \newcommand{\tq}{q}    
    
      \newcommand{\ta}{\tilde{a}}
    \newcommand{\tb}{\tilde{b}}
    \newcommand{\tc}{\tilde{C}}
    \newcommand{\hc}{\tilde{C}}
    \newcommand{\bC}{\bar{C}}
 \newcommand{\cb}{{c_{\beta}}}
 \newcommand{\bc}{\bar{c}}
 \newcommand{\bcb}{{\bar{c}_\beta}}
\newcommand{\icc}{\tilde{I}^{(c,\bc)}}
\newcommand{\ff}{{_{2}F_{1}}}
\newcommand{\cff}{{_1}F_{1}}
\newcommand{\nx}{e^{2iX/N}}
\newcommand{\ny}{e^{2iY/N}}
\newcommand{\calI}{\mathcal{I}}
\newcommand{\tA}{\tilde{A}^{(\tp+k,\tq)}}

\newcommand{\ot}{O\left(
\frac{1}{N^3}
\right)}
\newcommand{\bJ}{\bar{\mathcal{J}}}
 \reversemarginpar

\newcommand\numberthis{\addtocounter{equation}{1}\tag{\theequation}}

\begin{document}

\title[Circular Jacobi $\beta$-ensemble]{Asymptotic correlations with corrections for the circular Jacobi $\beta$-ensemble}
\author{Peter J. Forrester}
\address{School of Mathematical and Statistics, ARC Centre of Excellence for Mathematical and Statistical Frontiers, The University of Melbourne, Victoria 3010, Australia}
\email{pjforr@unimelb.edu.au}

\author{Shi-Hao Li}
\address{ School of Mathematical and Statistics, ARC Centre of Excellence for Mathematical and Statistical Frontiers, The University of Melbourne, Victoria 3010, Australia}
\email{lishihao@lsec.cc.ac.cn}

\author{Allan K. Trinh}
\address{ School of Mathematical and Statistics, ARC Centre of Excellence for Mathematical and Statistical Frontiers, The University of Melbourne, Victoria 3010, Australia}
\email{a.trinh4@student.unimelb.edu.au}

\subjclass[2010]{15B52, 15A15, 33E20}
\date{}

\dedicatory{}

\keywords{circular Jacobi $\beta$-ensemble;~confluent hypergeometric kernel; Routh-Romanovski polynomials; asymptotic expansion}

\begin{abstract}
Previous works have considered the leading correction term to the scaled limit of various correlation functions and distributions
for classical random matrix ensembles and their $\beta$ generalisations at the hard and soft edge. It has been found that the functional
form of this correction is given by a derivative operation applied to the leading term. In the present work we compute the leading 
correction term of the correlation kernel at the spectrum singularity for the circular Jacobi ensemble with Dyson indices $\beta = 1,2$ and
4, and also to the spectral density in the corresponding $\beta$-ensemble with $\beta$ even. 
The former requires an analysis involving the Routh-Romanovski polynomials, while the latter is based on multidimensional integral
formulas for generalised hypergeometric series based on Jack polynomials.
In all cases this correction term is
found to be related to the leading term by a derivative operation.

\end{abstract}
\maketitle

\section{Introduction}

The establishment of exact limit laws for the spectral statistics of large random matrices from
classical ensembles is a fundamental part of the theory of universality, whereby it is established that the limit laws depend only on low order moments of the matrix entries \cite{EY17}. Thus the limiting distributions exhibited by the classical ensembles are shared by a much wider class of random matrices, and most importantly have the distinguishing property that they can be computed exactly. The first result of this type was obtained by Wigner \cite{Wi55,Wi58}. In these works it was established that for random real symmetric matrices, upper triangular entries independent, with finite variance normalised to $1/2$ for convenience and bounded moments, the global spectral density $\bar{\rho}(\lambda)$  --- obtained by scaling the eigenvalues by dividing by $\sqrt{N}$ --- has the functional form $\rho_{\rm W, 0}(x):={1 \over 2 \pi} (1 - x^2)^{1/2}$, supported on $|x|<1$, nowadays  referred to as the Wigner semi-circle law. This implies that for a suitable class of test functions
$\phi(\lambda) $,
\begin{equation}\label{eq:2.1}
	\lim_{N\to\infty} \int_{-\infty}^\infty \phi(\lambda)\bar{\rho}(\lambda)\, d\lambda = \int_{-1}^{1} \phi(\lambda)\rho_{\rm W,0}(\lambda) \, \, d\lambda.
\end{equation}

The classical ensembles exhibiting the Wigner semi-circle law are the Gaussian orthogonal and unitary ensembles, with matrices $G$ constructed out of standard real and complex Gaussian random matrices $X$ according to $G = 
2^{-3/2}(X + X^\dagger)$ --- here the normalisation $2^{-3/2}$ ensures the variance of the upper triangular entries equals $1/2$. The special mathematical structures associated with the classical ensembles permit the determination of the rate of convergence to the Wigner law. For example, with $\phi = \phi_x$, where
$\phi_x(\lambda) = 1$ for $\lambda \le x$,  $\phi_x(\lambda) = 0$ otherwise, it has been proved by G\"otze and Tikhomirov \cite{GT05} for the GUE, and by Kholopov et al.~\cite{KTT08} in the case of the GOE, that there is a constant $C$ independent of $N$ such that
\begin{equation}\label{eq:2.1a}
\sup_{x} \Big |  \int_{-\infty}^\infty \phi_x(\lambda)( \bar{\rho}(\lambda) - \rho_{\rm W,0}(\lambda) ) \,d \lambda \Big | \le {C \over N},
\end{equation}
valid for all $N$. Most significantly, this $O(1/N)$ rate is expected to hold for the wider class of random matrices stated above (the Wigner class) in relation to (\ref{eq:2.1}).
In keeping with the convergence rate exhibited in (\ref{eq:2.1a}),
for $\phi(\lambda)$ in (\ref{eq:2.1}) smooth, large $N$ expansions of the form
\begin{equation}\label{eq:2.1b}
 \int_{-\infty}^\infty  \phi(\lambda)\bar{\rho}(\lambda)\, d\lambda = \int_{-1}^{1} \phi(\lambda)\rho_{\rm W,0}(\lambda) \, \, d\lambda +
{1 \over N} \int_{-\infty}^{\infty} \phi(\lambda)\rho_{\rm W,1}(\lambda) \, \, d\lambda + \cdots
\end{equation}
are well known in the context of the loop equation analysis of the Gaussian ensembles \cite{WF14}. Moreover (\ref{eq:2.1b}) can be extended to the Wigner class, with $\rho_{\rm W,1}(\lambda)$ now dependent on the fourth moment, but no higher moments; see the introduction to \cite{FT19c} for further discussion and references. The point to note then is that the rate of convergence in this setting holds true generally, while the density function at this order weakly depends on the details of the distribution of the entries.

An effect very similar, but even more structured, was observed by Edelman, Guionnet and P\'ech\'e \cite{EGP16} in relation to hard scaling for certain ensembles of positive
definite matrices $X^\dagger X$, where the entries of $X$ -- itself rectangular of size $n \times N$ -- are independent with zero mean. Since the spectral density is strictly zero for negative values, yet the eigenvalue density
is nonzero for $x>0$, the origin is referred to as a hard edge. Hard edge scaling refers to a rescaling of the eigenvalues so that the spacings between eigenvalues in the neighbourhood
of the origin is of order unity. The limiting state depends on whether the entries of the matrix are real or complex (typically labelled by the Dyson index $\beta=1$ and $\beta = 4$), and the parameter $a=n-N$ which
specifies the number of extra rows relative to columns in $X$. With $E_{N,\beta}(0;J;a)$ denoting the probability that the domain $J \subset \mathbb R^+$ contains no eigenvalues, it was conjectured in
\cite{EGP16} that
\begin{equation}\label{eq:2.1c}
E_{N,2}(0;(0,s/4N);a) = E_2^{\rm hard}(0;(0,s);a) + {a \over 2 N} s {d \over ds} E_2^{\rm hard}(0;(0,s);a) + O\Big ( {1 \over N^2} \Big ),
\end{equation}
where
$$
 E_\beta^{\rm hard}(0;(0,s);a) = \lim_{N \to \infty} E_{N,\beta}(0;(0,s/4N);a).
 $$
 Proofs were subsequently given in \cite{PS16,Bo16}. With $\Sigma$ a fixed positive definite matrix, Hachen et al.~\cite{HHN16} gave a generalisation of (\ref{eq:2.1c}) for the product $X^\dagger \Sigma X$, and found
 that the structured form of the $1/N$ correction persists but with $a$ renormalised. Such a modification was also found in \cite{EGP16} for the setting that the Gaussian entries $X$ are in convolution with a further
 distribution giving rise fourth moments differing from the Gaussian --- that latter deviation shows as an additive shift to $a$.
 A number of recent works \cite{FT19,MMM19,FL20} have focussed attention on generalising (\ref{eq:2.1c}), involving both the parameter $\beta$ and the details of the ensemble giving rise to the hard edge. For the classical
 $\beta$ ensembles with a hard edge -- namely the Laguerre and Jacobi weights -- it is found in all cases considered that the leading correction term is related to the limiting distribution by a derivative operation. 
 Moreover, this  effect is also exhibited at the soft edge of the Gaussian and Laguerre ensembles \cite{FFG06,JM12,Ma12,FT18,FT19a}, now with the leading order correction proportional to $1/N^{1/3}$.
 Thus these recent studies have revealed that the leading order correction to the limiting hard and soft edge distributions for a number of classical random matrix ensembles is structured, and that the rate of convergence
 is not dependent on the details of the ensemble. 
 
 In this work we contribute to the study of this effect by considering convergence to the limit for a singularity type distinct from the hard and soft edge --- namely the spectrum singularity
 \cite[\S 3.9]{Fo10}. This singularity is exhibited at infinity by the classical Cauchy $\beta$-ensemble,
 itself specified by the joint eigenvalue probability density function (PDF) proportional to  
\begin{equation} \label{CyBetaPDF}
\prod_{j=1}^N \omega_\beta(x) \prod_{1\leq j<k\leq N} |x_k-x_j|^\beta, \quad -\infty<x_j<\infty,
\end{equation} 
where 
\begin{equation} \label{cb}
\omega_\beta(x)=(1-ix_j)^{\cb}(1+ix_j)^{\bcb}, \qquad c_\beta = -\beta (N+p-1)/2 - 1 + iq.
\end{equation}
Here $p,q$ are fixed constants, and $\bcb$ denotes the complex conjugate of $\cb$. However the reason for the name ``spectrum singularity'' does not become apparent until applying
a particular fractional linear transformation to map the line to the unit circle
\begin{align} \label{sp}
x_j =  i {1 + e^{i \theta_j} \over 1 - e^{i \theta_j}}.
%e^{i\T_j}=\frac{1-ix_j}{1+ix_j},
\end{align}
Now the spectrum singularity at infinity presents itself at  $\T=0$ on the unit circle in a more literal sense,
with the resulting joint eigenvalue PDF proportional to
\begin{align} \label{cJBetaPDF}
\prod_{j=1}^N e^{-q\T_j}
|1 - e^{i\T_j}|^{\beta p}\prod_{1\leq j<k\leq N}
\abs{e^{i\T_k}-e^{i\T_j}}^\beta, \quad \T_j\in [0,2\pi).
\end{align}
Following terminology introduced in \cite[\S 3.9]{Fo10}, this is said to specify the generalised circular Jacobi $\beta$ ensemble.
The case $p=q=0$ is recognised as the circular $\beta$ ensemble; see e.g.~\cite[\S 2.8]{Fo10}. The case $q=0$, through the
factor $\prod_{j=1}^N  
|1 - e^{i\T_j}|^{\beta p}$, introduces   the spectrum
singularity at $\theta = 0$. The discontinuity created by the further factor $\prod_{j=1}^N e^{-q\T_j}$ in the case $q \ne 0$ 
extends the spectrum singularity to the class of singularities first considered by Fisher and Hartwig in the theory of asymptotics
of Toeplitz determinants \cite{FH68}.

The limiting distribution in the neighbourhood of the spectrum singularity was first carried out in the case $\beta = 2$ \cite{NS93}.
This is the simplest case as it gives rise to a determinantal point process, meaning that all $k$-point correlation functions are determined
by a single kernel function; see e.g.~\cite[Ch.~5]{Fo10}. The extension to the case with $q\ne 0$ was given in \cite{BO01}; see also
the introduction to \cite{DKV11} for further references. For $\beta = 1$ and 4 the correlations now have a Pfaffian structure, and are
fully determined by a $2 \times 2$ correlation kernel. This itself is fully determined by what we will term the Christoffel-Darboux kernel.
The scaled limit in the neighbourhood of the spectrum singularity has been given in \cite{FN01}. In the present work this is extended to the
case $q \ne 0$. Moreover, we find the leading correction to the scaled limit.

\begin{proposition} \label{Main1}
Let
\begin{equation} \label{zx}
z(X) = i\frac{1+\nx}{1-\nx}.
\end{equation}
For $\beta = 1,2$ and $4$, the Christoffel-Darboux kernel $S_{\beta,N}(x,y)$ associated with the PDF \eqref{cJBetaPDF} admits the large $N$ expansion
\begin{equation} \label{Main1_Expansion}
S_{N,\beta}(z(X),z(Y)) {dz \over dX} 
= 
\hat{K}_{\infty,\beta}^{(\tp,\tq;0)}(X,Y) + \frac{1}{N}\hat{L}_{1,\beta}^{(\tp,\tq;0)}(X,Y) + O\left( {1\over N^2} \right),
\end{equation}
where $\hat{K}_{\infty,\beta}^{(\tp,\tq;0)}(X,Y)$, $\hat{L}_{1,\beta}^{(\tp,\tq;0)}(X,Y)$ are specified in \eqref{Beta2_K}, \eqref{Beta2_L}  for $\beta =2$,
in (\ref{L11}) for $\beta = 1$, and in (\ref{L14}) for $\beta = 4$.
Moreover,
\begin{align} \label{Main1_Derivative}
\hat{L}_{1,\beta}^{(\tp,\tq;0)}(X,Y)
=
\tp\left(
X\frac{\partial}{\partial X}+Y\frac{\partial}{\partial Y}+1
\right)
\hat{K}_{\infty,\beta}^{(\tp,\tq;0)}(X,Y).
\end{align}
\end{proposition}

%The structure of large $N$ expansion \eqref{Main1_Expansion} and the simplicity of correction term \eqref{Main1_Derivative} as a first order derivative operation of the leading term is not unique to circular Jacobi ensembles. Indeed this has been observed at the hard edge of the Wishart-Laguerre unitary ensemble (see \cite{EGP16,Bo16,PS16,FT19}) with the precise derivative operation proportional to \eqref{Main1_Derivative}. Analogous results at hard edges of the Jacobi ensemble for some specific parameters was demonstrated in \cite{?}. The immediate consequence is that the term proportional to $1/N$ in \eqref{Main1_Expansion} can be trivially eliminated by simple Taylor series expansion by replacing the stereographic projection \eqref{zx} with the optimally tuned scaling 
%\begin{equation} \label{OScale1}
%i\frac{1+ \exp( {2iX \over N + \tp} ) }{1- \exp( {2iX \over N + \tp} )}
%\end{equation}
%which gives a convergence to the limiting kernel at the rate $O(1/N^2)$. In the case $\beta=2$, this effect can also independently verified with use of differential equations and Painlev\'e transcendents.

The structural form of the $1/N$ correction in \eqref{Main1_Derivative} as a simple first order derivative operation of the limiting kernel in \eqref{Main1_Derivative} is precisely as found for the analysis
of the rate of convergence to the hard and soft edge limits. It
 implies that the optimally tuned scaling 
\begin{equation}
z(X) \mapsto i\frac{1+ \exp( {2iX/( N + \tp)} ) }{1- \exp( {2iX/ (N + \tp)} )}
\end{equation}
gives a convergence at the rate $O(1/N^2)$. In the case of the circular ensembles $p=q=0$, it implies that the leading order correction term is
$O(1/N^2)$ without any need for a shift; see \cite{BBLM06,FM15,BFM17} for an application of this fact in the case $\beta = 2$ to the analysis
of the statistical properties of the large Riemann zeros.

Analysis beyond the classical values $\beta=1,2$ and $4$ is also possible in restricted circumstances.
In particular, exact evaluations of the spectral density of \eqref{cJBetaPDF}, $\rho_{N,\beta}(\T)$, for $\beta$ even are known in the framework of the Selberg integral theory and its associated special functions,
allowing for the derivation of the following large $N$ expansion.

\begin{proposition} \label{Main2}
\label{rhoExpand}
For $\beta$ even and $N$ large, we have
\begin{equation}
{1\over N} \rho_{N,\beta}\left(
{\theta \over N} \right)
=
\rho_{\infty,\beta}(\theta) + {1\over N} \tp {d\over d\theta} [\theta \rho_{\infty,\beta}(\theta)]
+ O\left(
{1\over N^2} \right),
\end{equation}
where
\begin{equation}
\rho_{\infty,\beta}(\theta) = \lim_{N\to\infty} {1\over N} \rho_{N,\beta}\left(
{\theta \over N} \right)
\end{equation}
is an explicit $\beta$-dimensional integral given by \eqref{RhoInf} below.
\end{proposition}

As in Proposition \ref{Main1}, the term proportional to $1/N$ is a simple derivative operation of the limiting density and therefore trivially eliminated by tuning the hard edge scaling to $x/(N+p)$.
Thus the leading non-trivial correction term is $O(1/N^2)$. The same prescription was also given in \cite{FT19} at the hard edge for the Laguerre $\beta$-ensemble for even $\beta$ and the associated parameter $a>-1$.

We introduce the finite $N$ formulas which enable the derivation of Propositions \ref{Main1} and \ref{Main2} in Section \ref{S2a}.
The actual asymptotics is carried out in Section \ref{Sec2}, with some of the more technical working deferred to Appendix A.

\section{Preliminaries}\label{S2a}
\subsection{Correlations for $\beta = 2$}

%The two intimately related ensembles concerning this paper are the Cauchy ensemble on the real line and the circular Jacobi ensemble whose eigenvalues are distributed on the unit circle with a spectrum singularity. In the case of unitary symmetry which corresponds to the Dyson index $\beta =2$, the Cauchy unitary ensemble (CyUE) is proportional to
%\begin{equation} \label{CyPDF}
%	\prod_{j=1}^N \omega_2(x_j) \prod_{1\leq j<k\leq N} \abs{x_k-x_j}^2, \quad -\infty<x_j<+\infty.
%\end{equation}
%where the weight function
%\begin{align} \label{CyWeight}
%	\omega_2(x) &= e^{-2V(x)} \nonumber
%	\\
%	&=(1-ix)^c(1+ix)^{\bc}
%\end{align}
%with
%\begin{align}
%\Re c = -N - p, \quad \text{ and } \Im c = q
%\end{align}
%for some fixed $p,q$. The PDF \eqref{CyWeight} is an example of a classical ensemble which has the feature that the derivative of $V(x)$ in \eqref{CyWeight} given by
%$$
%2V'(x) = { g(x) \over f(x) }
%$$
%is a rational function with $\deg f \leq 2$ and $\deg g \leq 1$.

The PDF \eqref{CyBetaPDF} with Dyson index $\beta = 2$   is referred to as the Cauchy unitary ensemble (CyUE) and 
has the special feature of being  an example of an orthogonal polynomial determinantal point process. This means that its general $k$-point correlation function is given by a $k \times k$ determinant with entries determined by a single kernel function $S_{n,2}(x,y)$ involving orthogonal polynomials, denoted $\icc_n(x)$, associated with the weight \eqref{cb}. Explicitly, the $k$-point correlation function is
\begin{align} \label{RkDet}
	R_{N,2}^{(k)}(x_1, \ldots, x_k) = \det[ S_{N,2}(x_m,x_n) ]_{m,n=1}^k,
\end{align}
where
\begin{align} \label{CyKernel} 
	S_{N,2}(x,y)&=(\omega_2(x)\omega_2(y))^{1/2}\sum_{k=0}^{N-1}\frac{1}{h_k}{\icc_k(x)\icc_k(y)} \nonumber
	\\
	&=\frac{(\omega_2(x)\omega_2(y))^{1/2}}{h_{N-1}}\frac{\icc_N(x)\icc_{N-1}(y)-\icc_N(y)\icc_{N-1}(x)}{x-y}.
\end{align}
The second equality in \eqref{CyKernel} is the result of the Christoffel-Darboux summation; see e..g.~\cite[Prop.~5.1.3]{Fo10}.

The orthogonal polynomials $\icc_n(x)$ in \eqref{CyKernel} are required to be monic and have the orthgonality
\begin{align} \label{RROrt}
\int_{-\infty}^\infty \omega_2(x) \icc_n(x)\icc_m(x)dx=h_n\delta_{n,m}
\end{align}
for some normalisation $h_n > 0$. The polynomials with this property 
are known as the Routh-Romanovski polynomials (see e.g.~the review \cite{RWAK07}).
They have the explicit hypergeometric form
\begin{equation}\label{hyper}
\icc_n(x) = (-2i)^n 
{
\Gamma(c+\bc+n+1) \Gamma(n+c+1) \over \Gamma(c+\bc+2n+1) \Gamma(c+1)
}
\,_2 F_1 \left(-n, n+1+c+\bc ; c+1 ; { 1-ix \over 2} \right),
\end{equation}
and furthermore the normalisation has the gamma function form
\begin{align} \label{hn}
h_n=2^{2n+2+c+\bc} \pi
\frac{\Gamma(n+1)\Gamma(-c-\bc-2n)\Gamma(-c-\bc-2n-1)}{\Gamma(-c-\bc-n)\Gamma(-c-n)\Gamma(-\bc-n)}.
\end{align}
It is furthermore true that
\eqref{hyper} can be viewed as certain monic Jacobi polynomials $i^{-n} P_n^{(c,\bc)}(ix)$ of degree $n$.

Bijectively related to the CyUE by the fractional linear transformation \eqref{sp} is the eigenvalue PDF of the circular Jacobi unitary ensemble (cJUE) given by \eqref{cJBetaPDF} with $\beta=2$. Its corresponding kernel function can therefore be attained from \eqref{CyKernel} under the same transformations. Consequently, local behavior near the spectrum singularity will involve the Routh-Romanovski polynomials \eqref{hyper} mapped onto the unit circle.
As detailed in \cite{FN01} in the case $q=0$, and to be extended to include the cases $q \ne 0$ in Section \ref{S3.1} below, its asymptotics are evaluated from the confluent limit formula
\begin{equation} \label{conflimit}
\lim_{n\to\infty} \ff(-n,b;c;t/n)=\cff(b;c;-t).
\end{equation}
This limit formula follows immediately from the series forms
\begin{equation} \label{conflimit1}
\,_2 F_1 \left(a ,b; c ; z \right) = 
\sum_{\alpha=0}^\infty { (a)_\alpha (b)_\alpha \over (c)_\alpha } { z^\alpha \over \alpha! }, \quad \,_1 F_1 \left(a ; c ; z \right) = 
\sum_{\alpha=0}^\infty { (a)_\alpha \over (c)_\alpha } { z^\alpha \over \alpha! },
\end{equation}
where $(n)_\al=\Gamma(n+\al) / \Gamma(n)$ denotes the Pochhammer symbol.

From \eqref{CyKernel} and with $z(X)$ given in \eqref{zx},
\begin{equation} \label{SNlim}
\lim_{N \to\infty} S_{N,2}( z(X) , z(Y) ) {dz(X) \over dX } = K_{\infty,2}^{(\tp,\tq;0)}( X , Y ) 
\end{equation}
where
\begin{equation} \label{KInf2}
K_{\infty,2}^{(\tp,\tq;0)}( X , Y ) 
\propto
e^{-i(X+Y)}
{ (XY)^{(\tp+1)} \over X^2(X-Y) }
\left( X{\tc}^{(p,q,1)}_0(X)\tc^{(p,q,0)}_0(Y)- (X\leftrightarrow Y)
\right)
\end{equation}
with
\begin{equation} \label{tc0}
\tc_0^{(p,q,k)}(X) = \cff(p+k-iq;2p+2k;2iX);
\end{equation}
hence the name confluent hypergeometric kernel used in \cite{DKV11}.
The explicit proportionality is given in  \eqref{Beta2_K}.
 In the case $q=0$, from the formula
\begin{equation}
\cff(p;2p;2iX) = \Gamma(p+1/2) \left( {x \over 2} \right)^{-p+1/2} e^{iX} J_{p-1/2}(X) 
\end{equation}
%the RHS of  \eqref{KInf2} simplifies to \cite{NS93}
$(X/Y) K_{\infty,2}^{(\tp,0;0)}( X , Y )$ simplifies to \cite{NS93}
\begin{equation} \label{JKer}
(XY)^{1/2} { J_{p+1/2}(X)J_{p-1/2}(Y)-J_{p-1/2}(X)J_{p+1/2}(Y) \over 2(X-Y) }.
\end{equation}
 In the special case when $p,q=0$, there is the absence of the singularity in \eqref{cJBetaPDF}, which is then recognised
 as the eigenvalue PDF for the circular unitary ensemble (CUE) of random Haar distributed unitary matrices;
 see e.g.~\cite{DF17}. Using the Bessel function formula
\begin{equation}
J_{1/2}(X) = \left( {2 \over \pi X} \right)^{1/2} \sin X,
\quad
J_{-1/2}(X) = \left( {2 \over \pi X} \right)^{1/2} \cos X
\end{equation}
in \eqref{JKer}, one recovers the sine kernel
$$
{\sin \pi(X-Y) \over \pi(X-Y)},
$$
well known for describing bulk statistics in the CUE.

\subsection{Correlations for $\beta = 1,4$}
From the viewpoint of classical random matrix theory the Cauchy ensemble \eqref{CyBetaPDF} for the values $\beta=1,4$ is associated  with orthogonal and symplectic symmetries in the corresponding matrix models. The point processes 
of the eigenvalues are Pfaffian; see e.g.~\cite[Ch.~6]{Fo10}. For such ensembles with classical weights such as the generalised Cauchy ensemble, the $k$-point correlation function admits a Pfaffian with entries involving summations of skew orthogonal polynomials. For the respective $\beta$ values $1,4$,  
\begin{align} \label{RkB14}
R^{(k)}_{N,\beta}(x_1,\cdots,x_k)
=\text{Pf}\left[
\begin{array}{cc}
S_{N,\beta}(x_i,x_j)&I_{N,\beta}(x_i,x_j)\\
D_{N,\beta}(x_i,x_j)&S_{N,\beta}(x_j,x_i)
\end{array}
\right]_{i,j=1,\cdots,k}.
\end{align}
The functions $S_{N,\beta}(x,y)$ in \eqref{RkB14} are related to the kernel \eqref{CyKernel} via the relations \cite{AFNM00}
\begin{multline}\label{s1n}
S_{N,1}(x,y)=
\sqrt{ {\omega_2(x) \over \omega_2(y)} }
\frac{\tilde{\omega}_1(y)}{\tilde{\omega}_1(x)}S_{N-1,2}(x,y)
\\
+\frac{1}{2}\gamma_{N-2}\tilde{\omega}_1(y)\icc_{N-1}(y)\int_{-\infty}^\infty \text{sgn}(x-t)\icc_{N-2}(t)\tilde{\omega}_1(t)dt
\quad \text{for } N \text{ even},
\end{multline}
and
\begin{align}\label{s4n}
\begin{aligned}
S_{N,4}(x,y)&=\frac{1}{2}\sqrt{
\frac{\omega_2(y)\tilde{\omega}_4(x)}{\tilde{\omega}_4(y)\omega_2(x)}
}
S_{2N,2}(x,y)\\
&\quad-\frac{1}{2}\gamma_{2N-1}\frac{\omega_2(y)}{\sqrt{\tilde{\omega}_4(y)}}\icc_{2N}(y)
\int_x^\infty \icc_{2N-1}(t)\frac{\omega_2(t)}{(\tilde{\omega}_4(t))^{1/2}}dt.
\end{aligned}
\end{align}
Other functions in \eqref{RkB14} are
\begin{align*}
I_{1,N}(x,y)&=-\int_x^y S_{1,N}(x,z)dz-\frac{1}{2}\text{sgn}(x-y),\quad
I_{4,N}(x,y)=-\int_x^y S_{4,N}(x,z)dz,
\\
D_{\beta,N}(x,y)&=\frac{\partial}{\partial x}S_{\beta,N}(x,y).
\end{align*}
In the above formulas, $\gamma_j=(p-1-j)/h_j$ with $h_j$ given by \eqref{hn},
%The corresponding weights
%$$
%\omega_1 (x) = e^{-V_1(x)} ,\quad 
%\omega_4(x) = e^{-2V_4(x) }
%$$
%are obtained from the formulas \cite{AFNM00}
%$$
%V_1 = V + {1/2}\log f, \quad V_4 = V -{1/2}\log f
%$$
%by identifying the Pearson pair $(f,g)=(1+x^2, 2(N+p)x-2q)$ of $\omega_2$ given in \eqref{CyWeight}. 
and the weights in \eqref{s1n}, \eqref{s4n} are obtained by identifying the Pearson pair $(f,g)=(1+x^2, -2(N+p)x+2q)$ of $\omega_2$ given in \eqref{cb}.
According to  \cite[eq. (2.17)]{AFNM00} we therefore have
\begin{align} 
\begin{aligned}\label{Weight14}
&\tilde{\omega}_1(x)\quad=\quad\sqrt{\omega_2(x)/f(x)}
\quad=(1-ix)^{(c-1)/2}(1+ix)^{(\bc-1)/2},
\\
&\tilde{\omega}_4(x)\quad=\quad\quad\omega_2(x)f(x)
\quad\,\,\,=(1-ix)^{c+1}(1+ix)^{\bc+1}.
\end{aligned}
\end{align}
By equaling the weights \eqref{Weight14} and \eqref{cb} requires that the parameter $c$ in \eqref{s1n}, \eqref{s4n} must be
$$
\Re c = -N-p \quad\text{ and }  \quad \Im c = 2q
$$ 
for $\beta=1$ and 
$$
\Re c = -2N-2p\quad \text{ and } \quad \Im c = q
$$ 
for $\beta=4$.
We remark that similar Pffafian formulas for the case $\beta=1$ and $N$ odd can be stated (see \cite[eq. (2.11)]{AFNM00}).

%Like in the unitary case $\beta=2$, the stereographic projection \eqref{sp} on the PDF \eqref{CyBetaPDF} gives a $\beta$ generalisation of the circular Jacobi ensemble with PDF proportional to
%\begin{align} \label{cJPDF}
%\prod_{j=1}^N e^{-\tilde{q}(\T_j + \pi)}|1-e^{i\T_j}|^{\beta\tilde{p}}
%\prod_{1\leq j<k\leq N}\abs{e^{i\T_k}-e^{i\T_j}}^\beta, \quad \T_j\in [0,2\pi).
%\end{align}
%with $\cb=-\beta(N+\tilde{p}-1)/2-1+i\tilde{q}$.

Explicit asymptotic forms under the scaling \eqref{zx} in the neighbourhood of the spectrum singularity case $q=0$ of the functions \eqref{s1n}, \eqref{s4n}  have the known limit \cite{FN01}
$$
\lim_{N\to\infty} S_{N,\beta}(z(X),z(Y)) {dz(X) \over dX} = S_{\infty,\beta}(X,Y)
$$
where
\begin{multline}
S_{\infty,1}(X,Y)
=S_{\infty,2}(X,Y) \rvert_{p \to p+1 }
+\pi {\Gamma(p/2+1) \over \Gamma(p/2+1/2) }  { J_{p+1/2}(Y) \over (2Y)^{1/2} }
\\
\times
\left( 1- 2^{1/2} {\Gamma(p/2+3/2) \over \Gamma(p/2+1) }  \int_0^{X} s^{-1/2}	J_{p+3/2}(s) \, ds \right)
\end{multline}
and
\begin{equation}
S_{\infty,4}(X,Y)
=S_{\infty,2}(2X,2Y) \rvert_{p \to 2p}
-\pi p {J_{2p-1/2}(2 Y) \over  Y^{1/2} }
\int_0^{X} s^{-1/2}	J_{2p+1/2}(2 s) \, ds. 
\end{equation}

\subsection{The density for even $\beta$}
For even $\beta$, as with other classical $\beta$-ensembles, the spectral density of the circular Jacobi ensemble can be expressed in terms of generalised hypergeometric functions based on Jack polynomials with $\beta$ variables \cite[Ch.~13]{Fo10}. These exact results are known in the framework of Selberg integral theory. They involve the generalised hypergeometric function
$$
\,_p F_q^{(\alpha)}(a_1,\ldots,a_p;b_1,\ldots,b_q;x_1,\ldots,x_m)
=
\sum_{k = (k_1,\ldots,k_m)}
{1 \over \abs{k}!}
{[a_1]_k^{(\alpha)} \cdots [a_p]_k^{(\alpha)} \over [b_1]_k^{(\alpha)} \cdots [b_q]_k^{(\alpha)} }
C_k^{(\alpha)}(x_1,\ldots,x_m)
$$
where, with $k$ a partition of no more than $m$ parts, 
$$
\abs{k} = \sum_{j=1}^m k_j, \quad [a]_k^{(\alpha)} = \prod_{j=1}^m {\Gamma(a-(j-1)/\alpha+k_j) \over  \Gamma(a-(j-1)/\alpha+) }
$$
and $C_k^{(\alpha)}(x_1,\ldots,x_m)$ denotes the renormalised Jack polynomials.

Upon normalisation the PDF of cJ$\beta$E \eqref{cJBetaPDF} can be written
\begin{multline}\label{M1}
P_{N,\beta} (\theta_1,\theta_2, \ldots, \theta_N)
:=
{1 \over (2\pi)^N M_{N}(\tp\beta/2+i\tq,\tp\beta/2-i\tq,\beta/2) }
\\
\times
\prod_{j=1}^N e^{-q(\T_j+\pi)}
|1-e^{i\T_j}|^{\beta p}\prod_{1\leq j<k\leq N}
\abs{e^{i\T_k}-e^{i\T_j}}^\beta, \quad \T_j\in [0,2\pi),
\end{multline}
where
\begin{multline}\label{Morris}
M_N(a,b,\lambda)
=
\int_{-1/2}^{1/2} d\theta_1 \cdots \int_{-1/2}^{1/2} d\theta_N
\prod_{j=1}^N e^{\pi i \theta_j (a-b)} |1+e^{2\pi i\theta_j}|^{a+b}
\prod_{1\leq j<k\leq N}|e^{2\pi i\theta_k}-e^{2\pi i\theta_j}|^{2\lambda}
\\
=
\prod_{j=0}^{N-1} {\Gamma(\lambda j+a+b+1)\Gamma(\lambda(j+1)+1) \over 
\Gamma(\lambda j+a+1)\Gamma(\lambda j+b+1)\Gamma(1+\lambda) }.
\end{multline}

For $\beta$ even, it can be shown that the spectral density 
\begin{equation}
\rho_{N,\beta}(\theta)
=
N \int_0^{2\pi} d\theta_2 \cdots \int_0^{2\pi} d\theta_N
P_{N,\beta} (\theta,\theta_2, \ldots, \theta_N),
\end{equation}
can be expressed as \cite[Proposition 13.1.2]{Fo10}
\begin{multline} \label{rhoNbeta}
\rho_{N+1,\beta}(\theta)
\\
= (N+1) {M_N[(\tp-1)\beta/2+i\tq,(\tp+1)\beta/2-i\tq,\beta/2] \over M_{N+1}(\tp\beta/2+i\tq,\tp\beta/2-i\tq,\beta/2) } 
e^{-\tq\theta + iN\beta \theta/2 + iN\beta\pi/2} \abs{1-e^{i\theta}}^{\tp\beta}
\\
\times
\,_2 F_1^{(\beta/2)}(-N, \tp+1-2i\tq/\beta; -N-\tp-2(1+i\tq)/\beta +2 ; (e^{-i\theta})^\beta ).
\end{multline}

Crucially, the hypergeometric function appearing in \eqref{rhoNbeta} has the $\beta$-dimensional integral representation given by \cite[Proposition 13.1.4]{Fo10}
\begin{multline} \label{2F14beta}
\,_2 F_1^{(\beta/2)} (-N, -\tb , 2(\beta-1)/\beta +\ta + 1; (e^{-i\theta})^\beta )
=
\\
{1 \over (2\pi)^\beta} {1 \over M_\beta(N+\ta,\tb,2/\beta) }
\int_{-\pi}^{\pi} d\theta_1 \cdots \int_{-\pi}^{\pi} d\theta_\beta
\prod_{j=1}^\beta e^{i \theta_j (\ta-\tb)/2} |1+e^{i\theta_j}|^{\ta+\tb} (1+(1-e^{-i\theta})e^{i\theta_j} )^N
\\
\times \prod_{1\leq j<k\leq \beta}|e^{i\theta_k}-e^{i\theta_j}|^{4/\beta}
\end{multline}
with
\begin{equation} \label{tab}
\ta = 2\tp + 2/\beta -1,\quad \tb=-\tp-1+2i\tq/\beta.
\end{equation}
This formula substituted in (\ref{rhoNbeta}) will form the starting point of our analysis.

\section{The asymptotic expansions} \label{Sec2}
\subsection{The correlation kernel for $\beta=2$}\label{S3.1}
In this section Proposition \ref{Main1} will be demonstrated for case $\beta=2$. This is accomplished by computing the large $N$ expansion of the corresponding correlation kernel up to the first correction term. 
We begin with the scaled correlation kernel \eqref{CyKernel} on the unit circle given by
\begin{align} \label{eq:s3.1}
S_{N,2}\left(
i\frac{1+\nx}{1-\nx},\,i\frac{1+\ny}{1-\ny}
\right),
\end{align}
with our specific interest being the scaled large $N$ form in the neighbourhood of the spectrum singularity $\theta = 0$. 
As the Routh-Romanovski polynomials \eqref{hyper} play a crucial role in the computation of \eqref{CyKernel}, it is imperative to first understand its large $N$ asymptotics. 
It follows from \eqref{hyper}, the explicit value of $\omega_2(x)$ from (\ref{cb}), and the hypergeometric function transformation formula \cite[Eq. (15.8.7)]{dlmf}
\begin{equation}\label{TF}
{}_2 F_1 (a,b;c;z) = {\Gamma(c)  \Gamma(c - a - b) \over \Gamma(c-a) \Gamma(c-b)} \, {}_2 F_1 (a,b;a+b+1-c;1-z) , \qquad a \in \mathbb Z_{\le 0},
\end{equation}
that for $k=0,1,\cdots,N$ and $p+k>0$, 
\begin{align} \label{hyperA}
\icc_{N-k}\left(
i\frac{1+\nx}{1-\nx}
\right)
=
\left(
\frac{2i}{1-\nx}
\right)^{N-k} 
\ff\left(
-N+k, p+k-iq;
2p+2k ; 
1-\nx
\right).
\end{align}
Introduce now the confluent function
\begin{align}\label{apq}
\tA(j;X)=\frac{(p+k-iq)_j}{(2p+2k)_j}\cff(p+k+j-i q;2p+2k+j;2iX),\quad p+k>0,\,j\geq0.
\end{align}
For subsequent use, note that it follows from properties of the confluent hypergeometric  function in \eqref{apq} that
\begin{align}\label{aeq1}
\frac{d}{dX}\tA(j;X)=2i\tA(j+1;X)
\end{align}
and
\begin{align}\label{aeq2}
2iX(\tA(1;X)-\tA(2;X))
=2(\tp+k)\tA(1;X)-(\tp+k-i\tq)\tA(0;X).
\end{align}
The equation \eqref{aeq1} follows directly from the derivative of \eqref{eq3} and \eqref{aeq2} is the result of the contiguous relations for Kummer's confluent functions (See for example, \cite[Sec. 13.3]{dlmf}).
Of immediate relevance is the fact that applying the formula \eqref{conflimit} to (\ref{apq}) gives the limit
\begin{equation} \label{ILimit}
\lim_{N\to \infty} 
\left(\frac{1-e^{2iX/N}}{2i}
\right)^{N-k}
\icc_{N-k}\left(
i\frac{1+\nx}{1-\nx}
\right)
=
\tc_0^{(p,q,k)}(X),
\end{equation}
where $\tc_0^{(p,q,k)}(X) = \tA(0;X)$ is given in \eqref{tc0}.

Just knowing the limit \eqref{ILimit} is not sufficient to acquire the $1/N$ expansion to the \eqref{eq:s3.1}. An extension to the limits \eqref{conflimit} and \eqref{ILimit} to the first correction is also required. This can be achieved by studying factorial terms appearing from the series representation of \eqref{apq}. In particular, for large $N$ we have
\begin{multline}\label{poc}
(-N+k)_\al
=(-1)^\al N^\al
\bigg(
1-\frac{\al(2k+\al-1)}{2N}
\\
+\frac{\al(\al-1)}{24N^2}\left(3\al^2+(12k-7)\al+(12k^2-12k+2)\right)
+O\left(\frac{1}{N^3}\right)
\bigg)
\end{multline}
valid for fixed $\al\in\mathbb{N}$ and $k\in\mathbb{C}$.  This follows as a special case of the asymptotic formula \cite{TE51} 
\begin{multline}\label{rog}
\frac{\Gamma(z+a)}{\Gamma(z+b)}
=z^{a-b}
\bigg(
1+\frac{1}{2z}(a-b)(a+b-1)
\\
+\frac{1}{12z^2}{{a-b}\choose 2}(3(a+b-1)^2-a+b-1)
+O\left(\frac{1}{z^3}\right)
\bigg)
\end{multline}
when $|z|\to\infty$ and $|\text{arg}\,z|<\pi$, and has been used in the context of studying correction terms to limiting
kernels in random matrix theory in the recent work \cite{FL20}. Plugging the expansion \eqref{poc} in \eqref{hyperA}, we get
(see Appendix \ref{AppendixA1} for the details)
\begin{multline} \label{iccExpand}
\left(\frac{1-e^{2iX/N}}{2i}
\right)^{N-k}
\omega_2\left(\frac{1-e^{2iX/N}}{2i}
\right)
\icc_{N-k}\left(
i\frac{1+\nx}{1-\nx}
\right)
\\
=
\tc_0^{(p,q,k)}(X)
 + {1\over N} \tc_1^{(p,q,k)}(X) + O\left( {1\over N^2} \right),
\end{multline}
where $\tc_0^{(p,q,k)}$ is given in \eqref{tc0} and
\begin{multline}
\tc_1^{(p,q,k)}(X)
=
\frac{1}{2}(2iX)^2\left(\tA(1;X)-\tA(2;X)\right)
\\
-k(2iX)\tA(1;X)+(ik+q)XC_0^{(p,q,k)}(X).
\end{multline}

The term appearing in the numerator  from the Christoffel-Darboux summation in \eqref{CyKernel} prompts the introduction of the functions
\begin{align}
J^{(\tp,\tq,k)}_0(X,Y)
&= X{\tc}^{(p,q,k+1)}_0(X)\tc^{(p,q,k)}_0(Y)- (X\leftrightarrow Y), \label{J0}
\\
J^{(\tp,\tq,k)}_1(X,Y)
&= X\left({\tc}^{(p,q,k+1)}_0(X)\tc^{(p,q,k)}_1(Y)+{\tc}^{(p,q,k+1)}_1(X)\tc^{(p,q,k)}_0(Y)\right)- (X\leftrightarrow Y) \label{J1}
\end{align}
and
\begin{align}\label{Q1}
Q_1^{(\tp,\tq,k)} =\tp(2p+2k+1).
\end{align}
Direct substitution of \eqref{iccExpand} into \eqref{CyKernel} (again, for details see Appendix \ref{AppendixA1}) results in the expansion seen in \eqref{Main1_Expansion} for the case $\beta=2$ with
\begin{align} \label{Beta2_K}
&\hat{K}_{\infty,2}^{(\tp,\tq;k)}(X,Y)
=
\frac{ 2^{2(\tp+k) }
|\Gamma(\tp+k+1-i\tq)|^2}{\pi\Gamma(2\tp+2k+2)\Gamma(2\tp+2k+1)}
e^{-i(X+Y)-\tq\pi}
{ (XY)^{(\tp+k+1)} \over X^2(X-Y) }
J_0^{(\tp,\tq,k)}(X,Y)
\end{align}
and
\begin{multline} \label{Beta2_L}
\hat{L}_{1,2}^{(\tp,\tq;k)}(X,Y)
= \\
\frac{ 2^{2(\tp+k) }
|\Gamma(\tp+k+1-i\tq)|^2}{\pi\Gamma(2\tp+2k+2)\Gamma(2\tp+2k+1)}
e^{-i(X+Y)-\tq\pi}
{ (XY)^{(\tp+1)} \over X^2(X-Y) }
(J_1^{(\tp,\tq,k)}(X,Y)+Q_1^{(\tp,\tq,k)}J_0^{(\tp,\tq,k)}(X,Y)).
\end{multline}

Moreover, upon simplification \eqref{Beta2_L} is seen to be a simple directive operation of \eqref{Beta2_K}, as given by (\ref{Main1_Derivative}), which is stated as the following proposition.
\begin{proposition}\label{prop:3.2}
In the $\beta=2$ case, we have
\begin{align}
\hat{L}_{1,2}^{(\tp,\tq;0)}(X,Y)
=
\tp\left(
X\frac{\partial}{\partial X}+Y\frac{\partial}{\partial Y}+1
\right)
\hat{K}_{\infty,2}^{(\tp,\tq;0)}(X,Y).
\end{align}
\end{proposition}
\begin{proof}
To verify this equation, we must show
\begin{align}
\begin{aligned} \label{k10a}
J_1^{(\tp,\tq,0)}(X,Y)&+Q_1^{(\tp,\tq,0)}J_0^{(\tp,\tq,0)}(X,Y)\\
&=\tp
\left(-iX-iY+2\tp+X\frac{\partial}{ \partial X}+Y\frac{\partial}{\partial Y}\right)J_0^{(\tp,\tq,0)}(X,Y).
\end{aligned}
\end{align}
Now, from the contiguous formula \eqref{aeq2} for the confluent hypergeometric function, one can find 
\begin{align*}
\tc_1^{(\tp,\tq,k)}(X)&=\tp\cdot 2iX\tA(1;X)-\tp\cdot iX\tA(0;X).
\end{align*}
Thus, the left hand side of \eqref{k10a} is given by
\begin{multline} \label{k10b}
\tp\left(2iXY \tilde{A}^{(\tp+1,\tq)}(0;X) \tilde{A}^{(\tp,\tq)}(1;Y)+2iX^2 \tilde{A}^{(\tp+1,\tq)}(1;X) \tilde{A}^{(\tp,\tq)}(0,Y)\right)\\
\qquad\qquad- \tp\left(
iXY+iX^2-(2\tp+1)X
\right) \tilde{A}^{(\tp+1,\tq)}(0;X) \tilde{A}^{(\tp,\tq)}(0;Y)- (X\leftrightarrow Y).
\end{multline}
Whereas the right hand side \eqref{k10a} coincides with \eqref{k10b} with use of the derivative formula \eqref{aeq1}. Therefore the large $N$ expansion of \eqref{eq:s3.1} is
\begin{align}
\begin{aligned} \label{s2nExpand}
S_{N,2}( z(X), z(Y) ) {dz \over dX}
&=
\hat{K}_{\infty,2}^{(\tp,\tq;0)}(X,Y)
\\&\quad+ {1\over N} \tp\left(
X\frac{\partial}{\partial X}+Y\frac{\partial}{\partial Y}+1
\right)
\hat{K}_{\infty,2}^{(\tp,\tq;0)}(X,Y)
+ O\left( {1\over N^2} \right).
\end{aligned}
\end{align}
\end{proof}
Now we prepare to proceed to different symmetries with regards to the random matrix ensemble.
For the cases $\beta=1,4$, computation of the $k$-point correlation function at the hard regime for the circular Jacobi ensemble largely depends on the large $N$ behaviour of kernel
$$
S_{N,\beta}\left(
i\frac{1+\nx}{1-\nx},i\frac{1+\ny}{1-\ny}
\right)
$$
which can be explicitly derived from \eqref{s1n} and \eqref{s4n}. As a prerequisite, knowledge of the large $N$ expansion of $S_{N,2}$ is given in \eqref{s2nExpand} and the remaining terms involving some integrals in  \eqref{s1n}, \eqref{s4n} will be deferred to the Appendices \ref{AppendixA2} and \ref{AppendixA3}.

\subsection{The correlation kernel for $\beta=1$} \label{Sec3.2}

The following variables and functions are relevant to the case $\beta=1$. Denote
\begin{equation}
\eta_1=2\sqrt{\pi}\frac{\Gamma(\tp+2)\Gamma(\tp+5/2)}{|\Gamma((\tp+3)/2+i\tq)|^2},
\quad
\eta_2=-{(\tp+1)e^{-\tq\pi}} 2^{2p+2}
\frac{ |\Gamma(p+2-2iq)|^2 }{ \pi\Gamma(2\tp +4)\Gamma( 2\tp+5 ) }
\end{equation}
and introduce the integral operator
\begin{equation}
\mathcal{J}_o\left[
f(s)
\right](X)
=\int_0^X e^{-is-\tq\pi}s^{\tp+1}f(s)ds.
\end{equation}
For $N$ even and $z(X)$ be defined as in \eqref{zx}, direct large $N$ expansion by following formula \eqref{s1n} gives (for details, please refer to Appendix \ref{AppendixA2})
\begin{equation}\label{s1expand}
S_{N,1}\left(z(X),z(X)\right) {dz \over dX}
=
\hat{K}_{\infty,1}^{(\tp,2\tq;0)}(X,Y)+\frac{1}{N}\hat{L}_{1,1}^{(\tp,2\tq;0)}(X,Y)
+O\left(\frac{1}{N^2}\right)
\end{equation}
where
\begin{align*}
\hat{K}_{\infty,1}^{(\tp,2\tq;0)}(X,Y)
&=
{Y\over X} \hat{K}_{\infty,2}^{(\tp,2\tq;1)}(X,Y)
 +\frac{\eta_2}{X^2}e^{-iY}Y^{\tp+2}\tc_0^{(\tp,2\tq,1)}(Y)\left(
\mathcal{J}_{o}[\tc_0^{(\tp,2\tq,2)}(s)]-\frac{\eta_1}{2}
\right) 
\\
\hat{L}_{1,1}^{(\tp,2\tq;0)}(X,Y)
&=
{Y\over X} \hat{L}_{1,2}^{(\tp,2\tq;1)}(X,Y)
+\frac{\eta_2}{X^2}e^{-iY}Y^{\tp+2}\bigg(
\tc_0^{(\tp,2\tq,1)}(Y)\left(
\mathcal{J}_o[\tc_1^{(\tp,2\tq,1)}(s)]-\frac{\tp(\tp+2)\eta_1}{2}
\right)
\\
&\quad +\tc_1^{(p,2q,1)}(Y)\left(
\mathcal{J}_o[\tc_0^{(\tp,2\tq,1)}(s)]-\frac{\eta_1}{2}
\right)
\bigg) \numberthis \label{L11}
\end{align*}
with $\hat{K}_{\infty,2}^{(\tp,2\tq;1)}$, $\hat{L}_{1,2}^{(\tp,2\tq;1)}$ specified by \eqref{Beta2_K}, \eqref{Beta2_L}.
A direct computation gives the following proposition, as a special case of our main result in Proposition \ref{Main1}.
\begin{proposition}
For $\beta=1$, we have
\begin{align} 
\hat{L}_{1,1}^{(\tp,2\tq;0)}(X,Y)
=
\tp\left(
X\frac{\partial}{\partial X}+Y\frac{\partial}{\partial Y}+1
\right)
\hat{K}_{\infty,1}^{(\tp,2\tq;0)}(X,Y).
\end{align}
\end{proposition}
\begin{proof}
The proof is divided into three parts according to the number of terms in the kernels.
The first part is easily verified similar to the proof of Proposition \ref{prop:3.2}. The second term is to verify
\begin{align*}
&\hc_0^{(p,2q,1)}(Y)\mathcal{J}_o[\hc_1^{(p,2q,2)}(s)]+\hc_1^{(1)}(Y)\mathcal{J}_o[\hc_0^{(p,2q,2)}(s)]+\tp(2\tp+3)\hc_0^{(p,2q,1)}(Y)\mathcal{J}[\hc_0^{(p,2q,2)}(s)]\\
&\quad=\tp\left(\tp+1-iY+Y\frac{d}{dY}\right)(\hc_0^{(p,2q,1)}(Y)\mathcal{J}_o[\hc_0^{(p,2q,2)}(s)])+\tp\hc_0^{(p,2q,1)}(Y)X\frac{d}{dX}\mathcal{J}_o[\hc_0^{(p,2q,2)}(s)].
\end{align*}
By making use of the confluent hypergeometric function, one has
\begin{align}\label{ortho3}
\begin{aligned}
&\frac{d}{dY}\hc_0^{(p,2q,1)}(Y)=2i\tilde{A}^{(\tp+1,2\tq)}(1;Y),\\
&\hc_1^{(p,2q,1)}(Y)=2i\tp Y\tilde{A}^{(\tp+1,2\tq)}(1;Y)-i\tp Y\tilde{A}^{(\tp+1,2\tq)}(0;Y).
\end{aligned}
\end{align}
Substituting them into the above equation, the task becomes to verify that
\begin{align*}
\mathcal{J}_o[\hc_1^{(p,2q,2)}(s)]+\tp(\tp+2)\mathcal{J}_o[\hc_0^{(p,2q,2)}(s)]=\tp X\frac{d}{dX}\mathcal{J}_o[\hc_0^{(p,2q,2)}(s)]=\tp e^{-iX-\tq\pi}X^{\tp+2}\hc_0^{(p,2q,2)}(X).
\end{align*} 
Since when $X=0$, both sides of the equation are equal to zero, we can show the equation by taking the derivative on both sides again and one gets
\begin{align*}
\hc_1^{(p,2q,2)}(X)=-i\tp X\hc_0^{(p,2q,2)}(X)+\tp X\frac{d}{dX}\hc_0^{(p,2q,2)}(X),
\end{align*}
which in fact could be verified by recognising that
\begin{align*}
\hc_1^{(p,2q,2)}=2 i\tp X\tilde{A}^{(\tp+2;2\tq)}(1;X)-i\tp X\tilde{A}^{(\tp+2;2\tq)}(0;X), \quad \frac{d}{dX}\hc_0^{(p,2q,2)}(X)=2i\tilde{A}^{(\tp+2;2\tq)}(1;X).
\end{align*}
An equivalent form of the last term is to demonstrate
\begin{align*}
\hc_1^{(p,2q,1)}(Y)+\tp(\tp+1)\hc_0^{(p,2q,1)}(Y)=\tp\left(\tp+1-iY+Y\frac{d}{dY}\right)\hc_0^{(p,2q,1)}(Y),
\end{align*}
and it is established by using \eqref{ortho3}.
\end{proof}

Analogous expansion to \eqref{s1expand} for the case $N$ odd and $\beta=1$ too can be made explicit from the kernel \eqref{s1n} using the formula \cite[Eq. 6.112]{Fo10}  
\begin{multline}
S^\text{odd}_{N,1}(x,y) = S_{N,1}(x,y)\bigg\rvert_{N \to N-1}
+ \omega_1(y) {\icc_{N-1}(y) \over 2\tilde{s}_{N-1} }
\\
-{1\over 2} \gamma_{N-3} \tilde{s}_{N-3} {\omega_1(y) \over \tilde{s}_{N-1} }
\int_{-\infty}^\infty \omega_1(t)\text{sgn}(x-t) \left(
\icc_{N-1}(t)\icc_{N-2}(y)-\icc_{N-2}(t)\icc_{N-1}(y)
\right),
\end{multline}
where
\begin{align*}
\tilde{s}_k:=\frac{1}{2}\int_{-\infty}^\infty \omega_1(x)\icc_{k}(x)dx.
\end{align*}
This yields the same large $N$ expansion \eqref{s1expand} with identical leading and correction terms.

\subsection{The correlation kernel for $\beta=4$}
Here we define
\begin{align*}
\mathcal{J}_{s,0}&=\mathcal{J}_s\left[
\tc_0^{(2\tp,\tq,0)}(2Y)\tc_0^{(2\tp,\tq,1)}(2s)
\right],
\\
 \mathcal{J}_{s,1}&=\mathcal{J}_s\left[
\tc_1^{(2\tp,\tq,0)}(2Y)\tc_0^{(2\tp,\tq,1)}(2s)+\tc_0^{(2\tp,\tq,0)}(2Y)\tc_1^{(2\tp,\tq,1)}(2s)\right]+2\tp(4\tp+1)\mathcal{J}_{s,0},
\end{align*}
where the integral operator $\mathcal{J}_s$ is given by
\begin{align*}
\mathcal{J}_s[f(s)](X,Y)
=
\tp 
\frac{2^{4\tp} |\Gamma(2\tp+1-i\tq)|^2}{\pi\Gamma(4\tp+1)\Gamma(4\tp+2)}
e^{-\tq\pi-2iY}Y^{2\tp+1}\int_0^X e^{-2is}s^{2\tp}f(s)ds.
\end{align*}
For large $N$, we have
\begin{equation}\label{beta4Prop}
S_{N,4}\left(z(X),z(X)\right) {dz \over dX}
=
\hat{K}_{\infty,4}^{(2\tp,\tq;0)}(X,Y)
+\frac{1}{N}\hat{L}_{1,4}^{(2\tp,\tq;0)}(X,Y)
+O\left(\frac{1}{N^2}\right),
\end{equation}
where
\begin{align*}
\hat{K}_{\infty,4}^{(2\tp,\tq;0)}(X,Y)&=
\frac{Y}{2X}\hat{K}_{\infty,2}^{(2\tp,\tq;0)}(2X,2Y)+ {1\over X^2} \mathcal{J}_{s,0}
,
\\
\hat{L}_{1,4}^{(2\tp,\tq;0)}(X,Y)&=
\frac{Y}{2X}\hat{L}_{1,2}^{(2\tp,\tq;0)}(2X,2Y) 
+ {1\over X^2} \mathcal{J}_{s,1}.  \numberthis \label{L14}
\end{align*}
The term $O(1/N^2)$ term in (\ref{beta4Prop}), $\hat{L}_{2,4}^{(2\tp,\tq;0)}(X,Y)$ say, is made explicit in \eqref{L24} below.

%\hat{K}_{4,2}^{(2\tp,\tq;0)}(X,Y)&=
%\frac{Y}{2X}\left(
%K_2^{(2\tp,\tq;0)}(2X,2Y)+\frac{X^2-Y^2}{6}K_0^{(2\tp,\tq;0)}(2X,2Y)
%\right)
%+\mathcal{J}_{s,2}
%+\frac{1}{3}\left(
%\frac{XY}{2}K_0^{(2\tp,\tq;0)}(2X,2Y)+\mathcal{J}_{s,0}
%\right)

The correction terms \eqref{L11}, \eqref{L14} exhibit the structure \eqref{Main1_Derivative}, which means we can state the following proposition.
\begin{proposition}
For $\beta=4$ case, we have
\begin{align} 
\hat{L}_{1,4}^{(2\tp,\tq;0)}(X,Y)
=
2\tp\left(
X\frac{\partial}{\partial X}+Y\frac{\partial}{\partial Y}+1
\right)
\hat{K}_{\infty,4}^{(2\tp,\tq;0)}(X,Y).
\end{align}
\end{proposition}
\begin{proof}
 This will be demonstrated explicitly for the case $\beta=4$  by a direct derivative computation of $\hat{K}_{\infty,4}^{(2\tp,\tq;0)}(X,Y)$. Immediate from \eqref{Main1_Derivative} with $\beta=2$, we have
$$
\frac{Y}{2X}\hat{L}_{1,2}^{(2\tp,\tq;0)}(2X,2Y)=
2\tp\left(X\frac{\partial}{\partial X}+Y\frac{\partial }{ \partial Y}+1\right)\left(
\frac{Y}{2X}\hat{K}_{\infty,2}^{(2\tp,\tq;0)}(2X,2Y)
\right).
$$
Therefore it is sufficient to show
\begin{align}\label{js}
{1\over X^2}\mathcal{J}_{s,1}=2\tp\left(X\frac{\partial}{ \partial X}+Y\frac{\partial}{ \partial Y}+1\right)
{1\over X^2}\mathcal{J}_{s,0},
\end{align}
which is equivalent to
\begin{align*}
\bC_1^{(0)}(Y)\bJ_s\left[\bC_0^{(1)}(s)\right]&+\bC_0^{(0)}(Y)\bJ_s\left[\bC_1^{(1)}(s)\right]+2\tp(4\tp+1)\bC_0^{(0)}(Y)\bJ_s\left[\bC_0^{(1)}(s)\right]\\
&\quad=2\tp\left((2\tp-2iY)\bC_0^{(0)}(Y)\bJ_s\left[
\bC_0^{(1)}(s)\right]+2\tp Y\frac{d}{dY}\bC_0^{(0)}(Y)\bJ_s\left[\bC_0^{(1)}(s)
\right]\right)\\
&\quad\quad+2\tp e^{-2iX}X^{2\tp+1}\bC_0^{(0)}(Y)\bC_0^{(1)}(X)
\end{align*}
with 
$
\bJ_s\left[
f(s)
\right
]=\int_0^X e^{-2is}s^{2\tp}f(s)ds
$
and $\bC_j^{(k)}(X) = \tc_j^{(2\tp,\tq,k)}(2X)$ for $j=0,1,2$.
By substituting the formulas \eqref{aeq1}, \eqref{aeq2},
\begin{align}\label{symp2}
\begin{aligned}
&\frac{d}{dY}\bC_0^{(0)}(Y)=\frac{d}{dY}\tilde{A}^{(2\tp,\tq)}(0;2Y)=4i\tilde{A}^{(2\tp,\tq)}(1;2Y),\\
&\bC_1^{(0)}(Y)=8i\tp Y\tilde{A}^{(2\tp,\tq)}(1;2Y)-4i\tp Y\tilde{A}^{(2\tp,\tq)}(0;2Y)
\end{aligned}
\end{align}
into the above equality, it can shown that the difference of the derivatives of the LHS and RHS of \eqref{symp2} vanish. Additionally, since both the LHS and RHS of \eqref{symp2} vanish at $X=0$,  \eqref{js} can now be concluded.
\end{proof}

%we get the simplified equation
%\begin{align}\label{symp1}
%\bJ_s[\bC_1^{(1)}(s)]+2\tp(2\tp+1)\bJ_s[\bC_0^{(1)}(s)]=2\tp e^{-2iX}X^{2\tp+1}\bC_0^{(1)}(X).
%\end{align}
%If one can show the derivative on both sides is valid and the both sides are equal to zero when $X=0$, then the above equation is true. The derivative on both sides of \eqref{symp1} gives
%\begin{align*}
%\bC_1^{(1)}(X)+2\tp(2\tp+1)\bC_0^{(1)}(X)=-4\tp iX\bC_0^{(1)}(X)+2\tp(2\tp+1)\bC_0^{(1)}(X)+2\tp X\frac{d}{dX}\bC_0^{(1)}(X)
%\end{align*}
%and the equality are which can be verified by relations \eqref{symp2}.

\subsection{The density for $\beta \in  2\mathbb{N}$} \label{Sec4}

In (\ref{rhoNbeta}) with the substitution  \eqref{2F14beta},  the $\theta$ and $N$ dependence are factorised in the integrand.
 Thus at the spectrum singularity, the large $N$ asymptotics of the $\beta$-dimensional integral \eqref{2F14beta} can be obtained by Taylor expansion without the need of a saddle point analyses. Details for the large $N$ expansions of the normalisation terms, identifications of the leading large $N$ terms and their simplifications are given in Appendix \ref{AppendixB}. This working leads to Proposition \ref{rhoExpand}, and moreover shows (see also \cite{Li17})
\begin{equation} \label{RhoInf}
\rho_\infty(\theta) = C^{(\tp,\tq)}_\beta e^{i\beta\theta/2} \theta^{\tp\beta}
\,_1 F_1^{(\beta/2)} (p+1-2i\tq/\beta ; 2\tp+2; (-i\T)^\beta ),
\end{equation}
where the generalised hypergeometric function has the $\beta$-dimensional integral representation given by
\begin{multline}
\,_1 F_1^{(\beta/2)} (p+1-2i\tq/\beta ; 2\tp+2; (-i\T)^\beta )
=
\int_{-\pi}^{\pi} d\theta_1 \cdots \int_{-\pi}^{\pi} d\theta_\beta
\prod_{j=1}^\beta e^{i\theta_j (\ta-\tb)/2} |1+e^{i\theta_j}|^{\ta+\tb} 
e^{i\T e^{i\T_j} }
\\
\times
\prod_{1\leq j<k\leq N}|e^{i\theta_k}-e^{i\theta_j}|^{4/\beta}
\end{multline}
with $\ta,\tb$ are given by \eqref{tab} and 
\begin{equation} \label{Cbeta}
C^{(\tp,\tq)}_\beta
=
{ (\beta/2)^{\tp\beta} \over 2\pi }
{ \Gamma(1+\beta/2)\Gamma(\tp\beta/2+i\tq+1)\Gamma(\tp\beta/2-i\tq+1)
\over \Gamma(\tp\beta+\beta/2+1)\Gamma(\tp\beta+1) }.
\end{equation}
As a consequence, with the better tuned scaling given by $x/(N+p)$, for large $N$ we have
\begin{equation}
{1\over N + \tp} \rho_{N,\beta}\left(
{\theta \over N + \tp} \right)
=
\rho_{\infty,\beta}(\theta) 
+ O\left(
{1\over N^2} \right).
\end{equation}

%For negative integer $a$ or $b$ \cite[Prop 13.1.4]{Fo10} the formula
%$$
%\,_2 F_1^{(1/\alpha)}(a,b;c;(t)^m) 
%=
%{ \,_2 F_1^{(1/\alpha)}(a,b;a+b+1+\alpha(m-1)-c;(1-t)^m)  \over
%\,_2 F_1^{(1/\alpha)}(a,b;a+b+1+\alpha(m-1)-c;(1)^m)  }
%$$
%to

%\section{Conclusion}
%
%For the $\beta$ generalisations of the circular Jacobi ensemble studied in this paper, the rate of convergence of some spectral statistics at the hard edge can be optimally tuned to $O(1/N^2)$ by choosing a singular scaling given by $x/(N+p)$. This was achieved from the large $N$ asymptotics of the correlation functions for classical values $\beta=1,2$ and $4$ and the spectral densities for even $\beta$. 
%
%As discussed in the Introduction, such feature was found at the hard edge of the Laguerre $\beta$ ensemble for similar values of $\beta$ \cite{FT19}. Moreover, Edelman, Guionnet and P\'ech\'e \cite{EGP16} showed for the distribution function of the smallest singular square values of a random complex matrix with entries having zero mean, unit variance and finite kurtosis, one can similarly find a optimally tuned scaling variable.
%
%It is not entirely unthinkable to conjecture that this is applicable for any $\beta>0$ and independent to the details of the matrix entries.

\section*{Acknowledgements} 
This work is part of a research program supported by the Australian Research Council (ARC) through the ARC Centre of Excellence for Mathematical and Statistical frontiers (ACEMS). AKT would like to thank Mario Kieburg for many helpful discussions in preparation for this manuscript.

\appendix\label{appendixa}
\section*{Appendix A}
\renewcommand{\thesection}{A} 
\setcounter{equation}{0}
\subsection{Case: $\beta=2$}\label{AppendixA1}

The asymptotic expansion of \eqref{iccExpand} requires the expansion of each factor appearing on the right hand side. 
Note that \eqref{apq} has the power series in the form
\begin{align}\label{eq3}
(2iX)^j\tA(j;X)
=\sum_{\al=0}^\infty\frac{(\tp+k-i\tq)_\al}{(2\tp+2k)_\al}\frac{(2iX)^\al}{\al!}\left(\al\cdots(\al-j+1)\right).
\end{align}
From a simple Taylor expansion of the term $\left(\frac{1-e^{2iX/N}}{2i}
\right)^{N-k}$ and substituting  the formula \eqref{rog} into the series representation of \eqref{hyperA} gives
\begin{multline}\label{prop2}
\left(\frac{1-e^{2iX/N}}{2i}
\right)^{N-k}\icc_{N-k}\left(-\cot\frac{X}{N}\right)
= 
\tA(0;X)
\\
+ {1\over N} \frac{1}{2}(2iX)^2\left(\tA(1;X)-\tA(2;X)\right)-k(2iX)\tA(1;X)
+ O\left( {1\over N^2}\right).
\end{multline}

The expansion of $\omega_2$ begins by noting
\begin{align}\label{CyWeightA}
\omega_2(x):=(1-ix)^c(1+ix)^{\bc}=(1+x^2)^{-p}\exp(2q\,\arctan x),\quad c=-p+iq,\quad x\in\mathbb{R}.
\end{align}
Use of  the trigonometric formula
\begin{align*}
\arctan\left(
i\frac{1+\nx}{1-\nx}
\right)=\frac{i}{2}\log\frac{1+i\cot\frac{X}{N}}{1-i\cot\frac{X}{N}}, 
\quad
\cot\frac{X}{N}=i\frac{e^{iX/N}+e^{-iX/N}}{e^{iX/N}-e^{-iX/N}}
\end{align*}
tells us that
\begin{align*}
\arctan\left(
i\frac{1+\nx}{1-\nx}
\right)=\frac{i}{2}\log \left(-e^{-2iX/N}\right)=\frac{i}{2}\log(-1)+\frac{X}{N}.
\end{align*}
Thus we see from \eqref{CyWeightA} that
\begin{align} \label{WeightN}
\omega_2\left(
i\frac{1+\nx}{1-\nx}
\right)^{1/2}
=
\left( \sin\frac{X}{N}
\right)^{N+\tp}e^{\tq\left(\frac{X}{N}-\frac{\pi}{2}\right)}.
\end{align}
Combining \eqref{prop2} with \eqref{WeightN} yields \eqref{iccExpand}.

We now take account of the scaling by \eqref{zx}. The expansion of the normalisation \eqref{hn} appearing in \eqref{CyKernel} can be achieved by using the formula \eqref{rog}. We have
\begin{multline} \label{CyNormalisation}
\frac{1}{h_{N-k}}\bigg|_{p=N+\tp,\,q=\tq}=h^{(\tp+k,\tq)}N^{2\tp+2k-1} 
\\
\times \bigg (
1+ \frac{\tp(2\tp+2k-1)}{N}+\frac{(\tp+k-1)(2\tp+2k-1)(6\tp^2-\tp-k)}{6N^2}+\ot
\bigg ),
\end{multline}
where 
\begin{align*}
h^{(\tp+k,\tq)}=\frac{2^{2(\tp+k)-2}|\Gamma(\tp+k-i\tq)|^2}{\pi\Gamma(2\tp+2k)\Gamma(2\tp+2k-1)}.
\end{align*}
Finally,
\begin{align}\label{eq2}
\frac{\sin X/N\sin Y/N}{\sin(X-Y)/N}=\frac{XY}{N(X-Y)}\left(
1-\frac{XY}{3N^2}+\ot
\right).
\end{align}

Therefore, by substituting \eqref{iccExpand}, \eqref{CyNormalisation}, \eqref{eq2} into \eqref{CyKernel}, we get
\begin{multline}
S_{N-k,2}\left(
i\frac{1+\nx}{1-\nx},i\frac{1+\ny}{1-\ny}
\right)
= \\
\frac{e^{-i(X+Y)-\tq\pi}h^{(\tp+k+1,\tq)}(XY)^{\tp+k+1}}{N(X-Y)}
\bigg ((J^{(\tp,\tq,k)}_0
+\frac{1}{N}\left(
J_1^{(\tp,\tq,k)}+Q_1^{(\tp,\tq,k)}J_0^{(\tp,\tq,k)}
\right) \\
+\frac{1}{N^2}\left(
J_2^{(\tp,\tq,k)}+Q_1^{(\tp,\tq,k)}J_1^{(\tp,\tq,k)}+Q_2^{(\tp,\tq,k)}J_0^{(\tp,\tq,k)}
\right)+\ot \bigg ).
\end{multline}
Here $J_0,J_1,Q_1$ are given in \eqref{J0}, \eqref{J1}, \eqref{Q1},
\begin{equation} \label{Q2}
Q_2^{(\tp,\tq,k)}=-\frac{XY}{3}+\frac{(\tp+k)(2\tp+2k+1)(6\tp^2-\tp-k-1)}{6}
\end{equation}
and
\begin{multline} \label{J2}
J^{(\tp,\tq,k)}_2
=
 X \bigg ({\tc}^{(\tp,\tq,k+1)}_2(X)\tc^{(\tp,\tq,k)}_0(Y) \\ +{\tc}^{(\tp,\tq,k+1)}_1(X)\tc^{(\tp,\tq,k)}_1(Y)+{\tc}^{(\tp,\tq,k+1)}_0(X)\tc^{(\tp,\tq,k)}_2(Y)\bigg )- (X\leftrightarrow Y)
\end{multline}
with
\begin{multline}
\tc_2^{(\tp,\tq,k)}(X)
=
\frac{1}{8}(2iX)^4\left(
\tA(2;X)-2\tA(3;X)+\tA(4;X)\right)
\\
+\frac{1}{6}(2iX)^3\left(\tA(1;X)-3(k+1)\tA(2;X)+(3k+2)\tA(3;X)
\right)
\\
+\frac{1}{4}(2iX)^2\left(
-2k\tA(1;X)+2k(k+1)\tA(2;X)
\right)
\\
+ \tc_1^{(\tp,\tq,k)}(X)-\tc_0^{(\tp,\tq,k)}(X)
+\left(
\frac{(ik+\tq)^2}{2}-\frac{k+\tp}{6}
\right)X^2\tc_0^{(\tp,\tq,k)}(X).
\end{multline}
Together with the Jacobian as a result of the change of variables \eqref{zx}, this working shows the large $N$ expansion of \eqref{CyKernel} up to the first is given by \eqref{Main1_Expansion},
and furthermore shows that explicit form of the $O(1/N^2)$ term is given by
\begin{multline} \label{Beta2_L2}
\hat{L}_{2,2}^{(\tp,\tq;0)}(X,Y) \\ = f(X,Y)\left(J_2^{(\tp,\tq,0)}(X,Y) +Q_1^{(\tp,\tq,0)}J_1^{(\tp,\tq,0)}(X,Y) +\left(Q_2^{(\tp,\tq,0)}+\frac{1}{3X^2}\right)J_0^{(\tp,\tq,0)}(X,Y) \right),
\end{multline}
where $f(X,Y)$ is equal to the prefactor of $J_0{(\tp,\tq,0)}$ in \eqref{Beta2_K}.
The expressions for $Q_2^{(\tp,\tq,0)}$, $J_2^{(\tp,\tq,0)}(X,Y) $ are given in \eqref{Q2}, \eqref{J2}.

\subsection{Case: $\beta = 1$} \label{AppendixA2}

Recall that for $\beta=1$ we take $\Re c = -N-p$ and $\Im c = 2q$. For the case $N$ even, the integral term \eqref{s1n} evaluates as \cite[Eq. 6.113]{Fo10}
\begin{align} \label{a1}
\int_{-\infty}^\infty \text{sgn}(x-t)\icc_{N-2}(t)\omega_1(t)dt=2\int_{-\infty}^x \icc_{N-2}(t)\omega_1(t)dt
-\left(
\int_{-\infty}^\infty \omega_1(t)dt
\right)\prod_{j=0}^{N/2-2}\frac{\gamma_{2j}}{\gamma_{2j+1}}
\end{align}
where
\begin{align} \label{a2}
\prod_{j=0}^{N/2-2}\frac{\gamma_{2j}}{\gamma_{2j+1}}&=
\frac{\Gamma\left(
\frac{2\tp+5}{2}
\right)\Gamma\left(
\frac{\tp+2}{2}
\right)\Gamma\left(
\frac{\tp+3}{2}
\right)}{\sqrt{\pi}\left|\Gamma\left(
\frac{\tp+3+2i\tq}{2}
\right)\right|^2} \frac{\Gamma\left(
\frac{N-1}{2}
\right)\Gamma\left(
\frac{N+\tp+1}{2}
\right)\left|
\Gamma\left(
\frac{N+\tp+1+2i\tq}{2}
\right)\right|^2
}{\Gamma\left(
\frac{N+\tp}{2}
\right)\Gamma\left(
\frac{N+2\tp+3}{2}
\right)\left|
\Gamma\left(
\frac{N+\tp+1}{2}
\right)\right|^2
}.
\end{align}
Recall $\gamma_j$ given below \eqref{s4n}. The orthogonality property of the Routh-Romanovski polynomials \eqref{RROrt} implies
\begin{align} \label{w1int}
\int_{-\infty}^\infty \omega_1(t)dt
&=2^{1-p}\pi\frac{\Gamma(p)}{\Gamma\left(
\frac{p+1+iq}{2}
\right)\Gamma\left(
\frac{p+1-iq}{2}
\right)}
\\
&=
\sqrt{\pi}\frac{\Gamma\left(
\frac{N+\tp}{2}
\right)\Gamma\left(
\frac{N+\tp+1}{2}
\right)^2}{\Gamma\left(
\frac{N+\tp+1}{2}
\right)\Gamma\left(
\frac{N+\tp+1+2i\tq}{2}
\right)\Gamma\left(
\frac{N+\tp+1-2i\tq}{2}
\right)}.
\end{align}
The second line in \eqref{w1int} is the result of the duplication formula for the gamma function. By combining \eqref{a2}, \eqref{w1int} with \eqref{a1} and making use of \eqref{rog}, we have
\begin{align*}
\int_{-\infty}^\infty \omega_1(t)dt=\sqrt{\pi}\frac{\Gamma\left(
\frac{N+\tp}{2}
\right)\Gamma\left(
\frac{N+\tp+1}{2}
\right)^2}{\Gamma\left(
\frac{N+\tp+1}{2}
\right)\Gamma\left(
\frac{N+\tp+1+2i\tq}{2}
\right)\Gamma\left(
\frac{N+\tp+1-2i\tq}{2}
\right)}
\end{align*}
and hence
\begin{align*}
\int_{-\infty}^\infty \icc_{N-2}(t)\omega_1(t)dt&=
\frac{\Gamma\left(
\frac{\tp+2}{2}
\right)\Gamma\left(
\frac{2\tp+5}{2}
\right)\Gamma\left(
\frac{\tp+3}{2}
\right)}{\Gamma\left(
\frac{\tp+3+2i\tq}{2}
\right)\Gamma\left(
\frac{\tp+3-2i\tq}{2}
\right)}\frac{\Gamma\left(
\frac{N-1}{2}
\right)}{\Gamma\left(
\frac{N+2\tp+3}{2}
\right)}\\&
=\eta_1N^{-\tp-2}\left(
1-\frac{\tp(\tp+2)}{N}+\frac{(\tp+2)(\tp+3)(3\tp^2+\tp+1)}{6N^2}+\ot
\right).
\end{align*}
Moreover the expansion of the Routh-Romanovski polynomials as given by \eqref{prop2} gives
\begin{multline*}
\left.\omega_1\left(-\cot\frac{Y}{N}\right)\icc_{N-k}\left(-\cot\frac{Y}{N}\right)\right|_{p=N+\tp,\,q=2\tq} =\\
\left(\frac{Y}{N}\right)^{\tp+k+1}e^{-{\tq\pi}-iY}(-1)^{k}\ \left(\tc_0^{(\tp,2\tq,k)}(Y)+\frac{1}{N}\tc_1^{(\tp,2\tq,k)}(Y)+\frac{1}{N^2}\tc_2^{(\tp,2\tq,k)}(Y)+\ot\right)
\end{multline*}
and therefore after a change of variables the first term on the left hand side of \eqref{a1} becomes
\begin{multline*}
\int_{-\infty}^{-\cot (X/N)}\icc_{N-2}(t)\omega_1(t)dt
=\frac{1}{N^{\tp+2}}\int_0^X e^{-\tq\pi-is}s^{\tp+1} \\ \times
\left(
\tc_0^{(\tp,2\tq,2)}(s)+\frac{1}{N}\tc_1^{(\tp,2\tq,2)}(s)+\frac{1}{N^2}\left(\tc_2^{(\tp,2\tq,2)}(s)+\frac{s^2}{3}\hc_0^{(2)}(s)\right)+\ot
\right)ds.
\end{multline*}
Therefore the large $N$ expansion of the kernel \eqref{s1n} under the scaling \eqref{zx} results in \eqref{s1expand}.

\subsection{Case: $\beta = 4$} \label{AppendixA3}
The process for the asympotic expansion of the symplectic kernel \eqref{s4n} is similar. With $\Re c = -2N-2p$ and $\Im c = q$, from the known expansion \eqref{prop2} one can find
\begin{multline} \label{icc4}
\int_{-\cot (X/N)}^\infty\icc_{2N-1}(t)\frac{\omega_2(t)}{(\omega_4(t))^{1/2}}dt
=-\int_{-\infty}^{-\cot (X/N)} \icc_{2N-1}(t)\omega_1(t)dt
\\
=\frac{e^{-\frac{\tq\pi}{2}}}{N^{2\tp+1}}
\int_0^X e^{-2is}s^{2\tp}
\left(
\tc_0^{(2\tp,\tq,1)}(2s)+\frac{1}{N}\tc_1^{(2\tp,\tq,1)}(2s)+\frac{1}{N^2}\tc_2^{(2\tp,\tq,1)}(2s)+\ot
\right).
\end{multline}
Together with the expansion
\begin{align*}
\gamma_{2N-1}=\frac{2\tp}{h_{2N-1}}=2\tp h^{(2\tp+1,\tq)}N^{4\tp+1}\left(
1+\frac{2\tp(4\tp+1)}{N}+\frac{\tp(4\tp+1)(24\tp^2-2\tp-1)}{3N^2}+\ot
\right)
\end{align*}
which is the result of \eqref{rog}, \eqref{icc4} and \eqref{s4n}, this implies \eqref{beta4Prop}. Moreover, the term proportional to $1/N^2$ in \eqref{beta4Prop} is explicitly given by
\begin{equation} \label{L24}
\hat{L}_{2,4}^{(2\tp,\tq;0)}(X,Y)
=
\frac{Y}{2X}\hat{L}_{2,2}^{(2\tp,\tq;0)}(2X,2Y)
+
{X^2-Y^2\over 6} \frac{Y}{2X}\hat{K}_{\infty,2}^{(2\tp,\tq;0)}(2X,2Y)
+ {1\over X^2} \mathcal{J}_{s,2},
\end{equation}
where
\begin{multline}
\mathcal{J}_{s,2}=\mathcal{J}_s\left[
\bC_2^{(2\tp,\tq,0)}(2Y)\tc_0^{(2\tp,\tq,1)}(2s)+\tc_1^{(2\tp,\tq,0)}(2Y)\tc_1^{(2\tp,\tq,1)}(2s)+\tc_0^{(2\tp,\tq,0)}(2Y)\tc_2^{(2\tp,\tq,1)}(2s)
\right]
\\
+2\tp(4\tp+1)\mathcal{J}_{s,1}+\tp(4\tp+1)(24\tp^2-2\tp-1)\mathcal{J}_{s,0}.
\end{multline}

\subsection{Case: $\beta$ even} \label{AppendixB}

The large $N$ expansion of \eqref{rhoNbeta} comes in two parts -- the first is the hypergeometric function component and second is the normalisation terms expressed in terms of Morris integrals.

\smallskip
\noindent
{\it Expanding the hypergeometic function.}
The feasibility of expanding the functional part of \eqref{rhoNbeta} requires expressing the $N$-dimensional integral to one with $\beta$ dimensions via known duality transformations. For negative integer $a,b$, from the relation \cite[Proposition 13.1.7]{Fo10}
$$
\,_2 F_1^{(1/\alpha)}(a,b;c;(t)^m)
=
{ \,_2 F_1^{(1/\alpha)}(a,b;a+b+1+\alpha(m-1)-c;(1-t)^m) \over 
\,_2 F_1^{(1/\alpha)}(a,b;a+b+1+\alpha(m-1)-c;(1)^m)
},
$$
(in the case $m=1$ this reduces to (\ref{TF})) the hypergeometric function in \eqref{rhoNbeta} becomes
\begin{equation} \label{B1}
{ \,_2 F_1^{(\beta/2)}(-N,\tp+1-2i\tq/\beta;2\tp+2;(1-e^{-i\theta})^\beta) \over 
\,_2 F_1^{(\beta/2)}(-N,\tp+1-2i\tq/\beta;2\tp+2;(1)^\beta)
}.
\end{equation}
With the parameters $\ta,\tb$ defined in \eqref{tab}, the numerator in \eqref{B1} has the $\beta$-dimensional integral representation given by
\begin{multline} \label{B2}
I_{N+1}(\theta) :=
{1 \over (2\pi)^\beta M_\beta(\ta,\tb,2/\beta) }
\int_{-\pi}^{\pi} d\theta_1 \cdots \int_{-\pi}^{\pi} d\theta_\beta
\prod_{j=1}^\beta e^{i\theta_j (\ta-\tb)/2} |1+e^{i\theta_j}|^{\ta+\tb} 
(1+(1-e^{-i\theta})e^{i\theta_j} )^N
\\
\times
\prod_{1\leq j<k\leq N}|e^{i\theta_k}-e^{i\theta_j}|^{4/\beta}.
\end{multline}
The normalisations appearing in \eqref{B1}, \eqref{B2} becomes $M_\beta(N+\ta,\b,2/\beta)$ as a result of the Morris integral \eqref{Morris}.

For some chosen function $f = f(\T_1,\ldots,\T_\beta)$, define 
\begin{multline}\label{B3}
\calI [f](\T) =
\int_{-\pi}^{\pi} d\theta_1 \cdots \int_{-\pi}^{\pi} d\theta_\beta
f(\T_1,\ldots,\T_\beta)
\prod_{j=1}^\beta e^{i\theta_j (\ta-\tb)/2} |1+e^{i\theta_j}|^{\ta+\tb} 
e^{i\T e^{i\T_j} }
\\
\times
\prod_{1\leq j<k\leq N}|e^{i\theta_k}-e^{i\theta_j}|^{4/\beta}
\end{multline}
and set $I_\infty(\T) = \calI[1](\T)$.

For large $N$ and from simple Taylor series expansion,
\begin{equation} \label{B4}
I_{N+1}\left( {\theta \over N+1} \right)
=
I_\infty(\T) + {1\over N} \left\{
\left( -i\T + {\T^2 \over 2} \right) \calI \left[ \sum_{j=1}^\beta e^{i\T_j} \right](\T)
+ {\T^2 \over 2} \calI \left[ \sum_{j=1}^\beta e^{2i\T_j} \right](\T)
\right\}
+ O\left( {1\over N^2} \right).
\end{equation}
With $p\in \{1,\ldots,\beta\}$, note that
\begin{align*}
{\partial \over \partial \T_p} \prod_{j=1}^\beta e^{i\T e^{i\T_j} }
=
-\T e^{i\T_p} \prod_{j=1}^\beta e^{i\T e^{i\T_j} },
\quad
{\partial^2 \over \partial \T_p^2} \prod_{j=1}^\beta e^{i\T e^{i\T_j} }
=
(-\T e^{i\T_p}+\T^2 e^{2i\T_p}) \prod_{j=1}^\beta e^{i\T e^{i\T_j} }.
\end{align*}
Hence
\begin{multline} \label{B5}
\calI \left[ -\T \sum_{j=1}^\beta e^{i\T_j} \right](\T)
=
\int_{-\pi}^{\pi} d\theta_1 \cdots \int_{-\pi}^{\pi} d\theta_\beta
\prod_{j=1}^\beta e^{i\theta_j (\ta-\tb)/2} |1+e^{i\theta_j}|^{\ta+\tb} 
\\ \times
\left( \sum_{p=1}^\beta {\partial \over \partial \T_p} \prod_{j=1}^\beta e^{i\T e^{i\T_j} } \right)
\prod_{1\leq j<k\leq \beta}|e^{i\theta_k}-e^{i\theta_j}|^{4/\beta}
\\
= -i\ta \beta I_\infty(\T) + i(\ta-\tb) \calI \left[\sum_{p=1}^\beta {1 \over 1 + e^{i\T_p}} \right](\T).
\end{multline}
The second equality in \eqref{B5} is the result of integration by parts. Similarly,
\begin{multline} \label{B6}
\calI \left[ -i\T \sum_{j=1}^\beta e^{i\T_j} +\T^2 \sum_{j=1}^\beta e^{2i\T_j} \right](\T)
=
\int_{-\pi}^{\pi} d\theta_1 \cdots \int_{-\pi}^{\pi} d\theta_\beta
\prod_{j=1}^\beta e^{i\theta_j (\ta-\tb)/2} |1+e^{i\theta_j}|^{\ta+\tb} 
\\ 
\times
\left( \sum_{p=1}^\beta {\partial^2 \over \partial \T_p^2} \prod_{j=1}^\beta e^{i\T e^{i\T_j} } \right)
\prod_{1\leq j<k\leq \beta}|e^{i\theta_k}-e^{i\theta_j}|^{4/\beta}
\\
= -i(\ta+\tb) \beta \T I_\infty(\T) 
+ i(\ta+\tb) \T \calI \left[\sum_{p=1}^\beta {1 \over 1 + e^{i\T_p} } \right](\T)
+ i\ta \T \calI \left[\sum_{p=1}^\beta e^{i\T_p} \right](\T)
\\
+ i {2\over \beta} \T \calI \left[ \sum_{\underset{j\neq k}{j,k=1}}^\beta 
e^{i\T_j} \left( { e^{i\T_j} \over e^{i\T_k}-e^{i\T_j} } - { e^{-i\T_j} \over e^{-i\T_k}-e^{-i\T_j} } \right)
\right](\T).
\end{multline}
Using the fact that the sum in the last term \eqref{B6} is invariant when interchanging the indices $\theta_j$ and $\T_k$, it follows that
$$
\calI \left[ \sum_{\underset{j\neq k}{j,k=1}}^\beta 
e^{i\T_j} \left( { e^{i\T_j} \over e^{i\T_k}-e^{i\T_j} } - { e^{-i\T_j} \over e^{-i\T_k}-e^{-i\T_j} } \right)
\right](\T)
= (\beta-1) \calI \left[\sum_{p=1}^\beta e^{i\T_p} \right](\T).
$$
Combining \eqref{B5}, \eqref{B6} with \eqref{B4} and observing 
$$
i\T I_\infty'(\T) = -\T \calI \left[\sum_{p=1}^\beta e^{i\T_p} \right](\T),
$$
we get
\begin{equation}\label{B7}
I_{N+1}\left( {\T \over N+1} \right)
=
I_\infty (\T) + {1\over N} [(\tq+i\beta/2+i\tp\beta/2)\T I_\infty(\T) + \tp \T I_\infty'(\T)]
+ O\left( {1\over N^2} \right).
\end{equation}

\noindent
\smallskip
{\it  Expanding the Morris integral.}
As a corollary of the multiplication formula for the gamma function we have
$$
\prod_{j=0}^{\beta-1} \Gamma(2j/\beta + z)
=
(2\pi)^{\beta/2-1} (\beta/2)^{1-\beta(1+2z)/2} \Gamma(\beta z/2)\Gamma(\beta (1+z)/2).
$$ 
The Morris integral \eqref{Morris} specifying the normalisation \eqref{2F14beta} becomes
\begin{multline}
{1\over M_\beta(N+\ta,\tb,2/\beta)}
=
{ (2\pi)^{\beta/2-1} \over (\beta/2)^{3\beta/2-1} }
\left( \prod_{j=0}^{\beta-1}
{ \Gamma(1+2/\beta) \over \Gamma( 2(j+1)/\beta +1 ) }
\right)
\\
\times
{ \Gamma(\beta/2 -\tp\beta/2+i\tq)\Gamma(-\tp\beta/2+i\tq)\Gamma(N\beta/2+\tp+\beta/2+1)
\Gamma(N\beta/2+\tp\beta+1)
\over
\Gamma(\beta (N-1)/2 + \beta\tp/2 + i\tq +1)\Gamma(N\beta/2 + \beta\tp/2 + i\tq+1)
}.
\end{multline}
It follows from this that  the ratio of the Morris integrals appearing in \eqref{rhoNbeta} evaluates to
\begin{equation}
{ 
\Gamma(N\beta/2+\tp\beta/2 + i\tq +1) \Gamma(N\beta/2+(\tp-1)\beta/2 + i\tq +1)
\Gamma(\tp\beta/2-i\tq+1)\Gamma(1+\beta/2)
\over
\Gamma(\beta N/2 + \tp\beta  +1) \Gamma(N\beta/2 +\beta/2 +1) \Gamma( (\tp-1)\beta/2 + i\tq +1)
}.
\end{equation}
Hence with use of \eqref{rog} we get
\begin{multline} \label{BGamma}
{1 \over M_\beta(\tq,\tb,2/\beta)}
{M_N[(\tp-1)\beta/2+i\tq,(\tp+1)\beta/2-i\tq,\beta/2] \over M_N(\tp\beta/2+i\tq,\tp\beta/2-i\tq,\beta/2) } 
=
{ C_\beta^{(\tp,\tq)} \over M_\beta(\ta,\tb,2/\beta) }
{ \Gamma(N\beta/2 + p\beta + \beta/2 + 1) \over \Gamma(N\beta/2 + \beta/2 + 1)}
\\
=
{ C_\beta^{(\tp,\tq)} \over M_\beta(\ta,\tb,2/\beta) } \left( {N\beta \over 2} \right)^{\tp\beta}
\left[
1 + {1\over N}(\tp^2\beta+\tp\beta +\tp) + O\left( {1\over N^2}\right)
\right],
\end{multline}
where $C_\beta^{(\tp,\tq)}$ is given by \eqref{Cbeta}.

The $\beta$-dimensional integral $I_\infty(\T)$ given below \eqref{B3} can be expressed in terms of a generalised hypergeometric function \cite[Chapter 13 Q4(i)]{Fo10}
$$
I_\infty(\T) = (2\pi)^\beta M_\beta(\ta,\tb,2/\beta) 
\,_1 F_1^{(\beta/2)}(\tp+1-2i\tq/\beta;2\tp+2;(-i\T)^\beta ).
$$
and this together with \eqref{BGamma} and \eqref{B7} establishes Proposition \ref{rhoExpand}.

\small
\bibliographystyle{abbrv}

\end{document}